\documentclass{article}

\usepackage{amsmath,amsthm,proof,amssymb}
\newcommand{\LbU}{\mathbf{Lb}^{\boldsymbol{*}}_{\mathbf{1}}}
\newcommand{\LbS}{\mathbf{Lb}^{\boldsymbol{*}}}
\newcommand{\Lb}{\mathbf{Lb}}
\newcommand{\U}{\mathbf{1}}
\newcommand{\tran}{\tau}
\newcommand{\BS}{\mathop{\backslash}}
\newcommand{\SL}{\mathop{/}}
\newcommand{\PMod}{\langle\rangle}
\newcommand{\NMod}{[]^{-1}}

\theoremstyle{definition}\newtheorem{theorem}{Theorem}
\newtheorem{lemma}[theorem]{Lemma}

\begin{document}

\title{Eliminating the Unit Constant in the Lambek Calculus with Brackets}
\author{Stepan Kuznetsov}

\maketitle

\begin{abstract}
We present a translation of the Lambek calculus with brackets and the unit
constant, $\LbU$, into the Lambek calculus with brackets allowing empty antecedents,
but without the unit constant, $\LbS$. Using this translation, we extend previously known
results for $\LbS$ to $\LbU$: (1) languages generated by categorial
grammars based on the Lambek calculus with brackets are context-free (Kanazawa 2017);
(2)
the polynomial-time algorithm for deciding derivability
of bounded depth sequents (Kanovich {\em et al.}\ 2017).
\end{abstract}

\section{Introduction}

The Lambek calculus~\cite{Lambek58} was introduced for describing natural language syntax
by means of type-logical (categorial) grammars. Further research on type-logical grammar
showed that the original system proposed by Lambek appears to be insufficient to cover
intrinsic linguistic phenomena, and so various extensions and modifications of the Lambek
calculus were introduced (see, for example, books by Morrill~\cite{MorrillBook},
Moot and Retor\'{e} \cite{MootRetore} and others). One of
these extensions is the {\em Lambek calculus with brackets} introduced by Moortgat~\cite{Moortgat}.
While the original Lambek calculus is fully associative, brackets block associativity in
specified situations, thus disallowing derivations of ungrammatical phrases like
``the girl whom John loves Mary and Pete loves'' (see~\cite{MorrillBook} for a more detailed
analysis).

The original Lambek grammars generate precisely context-free languages, as shown by Pentus~\cite{PentusCFG}.
For the Lambek calculus extended with brackets, the context-free upper bound was claimed by J\"ager~\cite{Jaeger}.
Unfortunately, J\"ager's argument relied upon a lemma by Versmissen~\cite{Versmissen}, which was
afterwards shown to be incorrect~\cite{FaddaMorrill}\cite{FSCD}. 

Recently, however, Kanazawa~\cite{KanazawaArXiv}
returned to this question and presented a new proof, based on a insightful combination of J\"ager's ideas and
the original Pentus' approach. Kanazawa proved this result for two versions of the Lambek calculus with brackets:
not allowing empty antecedents (we denote it by $\Lb$) and allowing them ($\LbS$).

The third variant of the calculus, $\LbU$, is obtained from $\LbS$ by adding the {\em multiplicative unit constant,}
$\U$. As noticed by Kanazawa~\cite{KanazawaArXiv}, Pentus-style reasoning is not applicable to the case with
the unit constant (even without brackets). In Pentus' proof, the target context-free grammar rules are essentially
all derivable sequents of a bounded size. For $\Lb$ and $\LbS$, the set of these sequents is finite; however,
the measure of size used by Pentus ignores occurrences of $\U$, thus this set of sequents becomes infinite, and
doesn't yield a context-free grammar. A workaround for this issue (in the case without brackets) was presented
by Kuznetsov~\cite{FG2011}: the unit constant gets eliminated by a faithful translation of the calculus with
the unit into the original system. The same problem for the case with brackets ({\it i.e.,} whether $\LbU$-grammars
generate precisely context-free languages) was left as an open question in~\cite{KanazawaArXiv}.

In this paper, we extend the construction from~\cite{FG2011} to an embedding of $\LbU$ into $\LbS$.
Thus, we show that $\LbU$-grammars generate the same class of languages as $\LbS$-grammars---and, due
to Kanazawa~\cite{KanazawaArXiv}, it is the class of context-free languages.

The translation of $\LbU$ to $\LbS$ also has an algorithmic application. While the original Lambek calculus,
and, therefore, the Lambek calculus with brackets is NP-complete (Pentus~\cite{PentusNP}), for sequents of
bounded formula complexity and bracket nesting depth Kanovich {\it et al.}~\cite{FSCD}, extending Pentus~\cite{PentusPoly},
present a polynomial-time algorithm for deciding derivability in $\LbS$. Our translation generalises this algorithm
from $\LbS$ to $\LbU$.

\section{The Calculi $\LbS$ and $\LbU$}

In this section and further we use the notation of~\cite{FSCD}; the notation used 
in~\cite{KanazawaArXiv} is slightly different.

The syntax of the Lambek calculus with brackets is a bit more involved than a standard
sequential calculus. Formulae of $\LbS$ are recursively built from variables
($\mathrm{Var} = \{ p_1, p_2, \dots \}$) using three binary connectives,
$\BS$ (left division), $\SL$ (right division), $\cdot$ (product) and two unary ones,
$\PMod$ and $\NMod$. The binary connectives come from the original Lambek calculus~\cite{Lambek58};
as one can easily see, the system $\LbS$ defined below is a conservative extension of the
version of the Lambek calculus which allows empty antecedents~\cite{Lambek61}. The unary
connectives operate brackets that are used to introduce controlled non-associativity.
Sequents of $\LbS$ are expressions of the form $\Pi \to C$, where $C$ is a formula and
$\Pi$ is a {\em meta-formula.} Meta-formulae are built from formulae using comma (,) and
brackets ([...]). By definition, comma is associative, and the empty meta-formula $\Lambda$
is the unit object w.r.t.\ comma: $\Gamma,\Lambda$ and $\Lambda,\Gamma$ are both considered
graphically equal to $\Gamma$. By $\Delta(\Pi)$ we denote a meta-formula $\Delta$ with
a {\em designated occurrence} of a sub-meta-formula $\Pi$.

Axioms and rules of $\LbS$ are as follows.

$$
\infer[(\mathrm{ax})]{p_i \to p_i}{}
$$
$$
\infer[(\BS\to)]{\Delta(\Pi, A \BS B) \to C}{\Pi \to A & \Delta(B) \to C}
\qquad
\infer[(\to\BS)]{\Pi \to A \BS B}{A, \Pi \to B}
$$
$$
\infer[(\SL\to)]{\Delta(B \SL A, \Pi) \to C}{\Pi \to A & \Delta(B) \to C}
\qquad
\infer[(\to\SL)]{\Pi \to B \SL A}{\Pi, A \to B}
$$
$$
\infer[(\cdot\to)]{\Delta(A \cdot B) \to C}{\Delta(A, B) \to C}
\qquad
\infer[(\to\cdot)]{\Pi, \Gamma \to A \cdot B}{\Pi \to A & \Gamma \to B}
$$
$$
\infer[(\PMod\to)]{\Delta(\PMod A) \to C}{\Delta([A]) \to C}
\qquad
\infer[(\to\PMod)]{[\Pi] \to \PMod A}{\Pi \to A}
$$
$$
\infer[(\NMod\to)]{\Delta([\NMod A]) \to C}{\Delta(A) \to C}
\qquad
\infer[(\to\NMod)]{\Pi \to \NMod A}{[\Pi] \to A}
$$

Adding the unit constant, $\U$, with the following rules of inference (cf.~\cite{Lambek69}),
$$
\infer[(\U\to)]{\Delta(\U) \to C}{\Delta(\Lambda) \to C}
\qquad
\infer[(\to\U)]{\to\U}{}
$$
yields the calculus $\LbU$. 

The cut rule of the following form is admissible~\cite{Moortgat}.
$$
\infer[(\mathrm{cut})]{\Delta(\Pi) \to C}{\Pi \to A & \Delta(A) \to C}
$$
Admissibility of cut allows replacing a subformula with an equivalent one preserving
derivability. 

\section{Translating $\LbU$ to $\LbS$}

In this section we present a translation of $\LbU$ formulae that eliminates the unit constant.

Informally, we replace each occurrence of $\U$ by $(q \BS q)$, where
$q$ is a fresh variable. For classical propositional logic, this would be sufficient, since
$(\varphi \Rightarrow \varphi)$ there is equivalent to the ``{\sc true}'' constant. In order to make
this construction work for the substructural
system $\LbU$, however, we also need to add some extra $(q \BS q)$'s,
depending on the polarity of specific subformula.

Formally, we define two translations, $\tran^+$ and $\tran^-$, by joint induction:
\begin{align*}
& \tran^+(\U) = \tran^-(\U) = q \BS q\\
& \tran^+(p_i) = (q \BS q) \cdot p_i \cdot (q \BS q) && \tran^-(p_i) = p_i\\
& \tran^+(A \BS B) = \tran^-(A) \BS \tran^+(B) && \tran^-(A \BS B) = \tran^-(A) \BS \tran^-(B) \\
& \tran^+(B \SL A) = \tran^+(B) \SL \tran^+(A) && \tran^-(B \SL A) = \tran^-(B) \SL \tran^+(A)\\
& \tran^+(A \cdot B) = \tran^+(A) \cdot \tran^+(B) && \tran^-(A \cdot B) = \tran^-(A) \cdot \tran^-(B)\\
& \tran^+(\PMod A) = (q \BS q) \cdot \PMod \tran^+(A) \cdot (q \BS q) &&
\tran^-(\PMod A) = \PMod \tran^-(A)\\
& \tran^+(\NMod A) = \NMod \tran^-(A) && 
\tran^-(\NMod A) = (q \BS q) \BS \NMod \tran^-(A) \SL (q \BS q)
\end{align*}

For metaformulae, we define only $\tran^-$:
\begin{align*}
&\tran^-(\Lambda) = \Lambda\\
&\tran^-(\Gamma, \Delta) = \tran^-(\Gamma), \tran^-(\Delta)\\
&\tran^-([\Gamma]) = [\tran^-(\Gamma)]
\end{align*}



\begin{theorem}\label{Th:main}
For any $\LbU$-sequent $\Pi \to C$ that has no occurrences of $q$,
$\Pi \to C$ is derivable in $\LbU$ if{f} $\tran^-(\Pi) \to \tran^+(B)$ is
derivable in $\LbS$.
\end{theorem}

In order to make the proof more convenient, we first reformulate $\LbU$: remove $(\U\to)$ and replace some of the other rules with the following ones
(here and further $\U^n$ means $\underbrace{\U, \U, \dots, \U}_{\text{$n$ times}}$).
$$
\infer[(\to\U)']
{\U^k \to \U}{}
\qquad
\infer[(\mathrm{ax})']
{\U^k, p_i, \U^m \to p_i}{}
$$
$$
\infer[(\to\PMod)']
{\U^k, [\Pi], \U^m \to \PMod A}
{\Pi \to \PMod A}
\qquad
\infer[(\NMod\to)']
{\Delta([\U^k, \NMod B, \U^m]) \to C}
{\Delta(B) \to C}
$$
All other rules are left intact.

Let's call the new calculus ${\LbU}'$.

\begin{lemma}\label{Lm:LbUp}
A sequent is derivable in ${\LbU}'$ if{f} it is derivable in 
$\LbU$.
\end{lemma}

\begin{proof}
Since the new $'$-rules become the original rules of $\LbU$ when $k=m=0$, for the
``if'' part it is sufficient to show that $(\U\to)$ is admissible in ${\LbU}'$, {\em i.e.,}
that if $\Psi(\Lambda) \to C$ is derivable in ${\LbU}'$, then so is $\Psi(\U) \to C$.

Proceed by induction on derivation.
Let's denote $\Pi$, $\Gamma$, and $\Delta$ (except for the designated part that is affected by the rule) 
in the rules of $\LbU$ as the {\em context}. Consider the last rule that derives $\Psi(\Lambda) \to C$.

If this designated occurrence of $\Lambda$ is inside the context of the last rule that derives $\Psi(\Lambda) \to C$,
then we can trace it upwards to the premise (one of the premises) of the rule, replace $\Lambda$ by $\U$ there
(the sequent is still derivable by induction hypothesis), and then apply the rule. In other words, in this case
the two rules are exchangeable, and we can propagate the $(\U\to)$ rule upwards. For $(\SL\to)$, $(\BS\to)$,
$(\to\SL)$, $(\to\BS)$, $(\cdot\to)$, $(\to\cdot)$, $(\PMod\to)$, and $(\to\NMod)$ this is the only possible
situation: a possible place for inserting $\U$ is always in the context.

If $\Psi(\Lambda) \to C$ is an axiom of the form $(\mathrm{ax})'$, namely, $\U^k, p_i, \U^m \to p_i$, then $\Psi(\U) \to C$ is either 
$\U^{k+1}, p_i, \U^m \to p_i$ or $\U^k, p_i, \U^{m+1} \to p_i$; both are again instances of $(\mathrm{ax})'$. The $(\to\U)'$ case
is handled in the same way. For the $(\to\PMod)'$ case, the only interesting situation is when $\U$ is added outside
the brackets (the case when it is added into the context $\Pi$, was already considered); in this situation the extra $\U$ again
gets absorbed by the $(\to\PMod)'$ rule. The $(\NMod\to)'$ case is dual.


For the ``only if'' part, just notice that each $'$-rule can be represented as a
consequent application of the corresponding rule of $\LbU$ and then $k+m$ times
$(\U\to)$  (or $k$ times for $(\to\U)'$). 
\end{proof}

\begin{proof}[Proof of Theorem~\ref{Th:main}]
The ``only if'' part: let $\Pi \to C$ be derivable in $\LbU$. Then by Lemma~\ref{Lm:LbUp}
it is derivable in ${\LbU}'$. We need to show that $\tran^-(\Pi) \to \tran^+(C)$ is derivable
in $\LbS$. Proceed by induction on the derivation of $\Pi\to C$ in ${\LbU}'$. Applications of rules
without $'$ are translated into $\LbS$ ``as is.'' The $(\to\U)'$ transforms into $(q\BS q)^k \to q \BS q$,
which is derivable in $\LbS$; $(\mathrm{ax})'$ becomes $(q \BS q)^k, p_i, (q \BS q)^m \to (q \BS q) \cdot p_i \cdot (q \BS q)$,
which is also derivable. Finally, $(\to\PMod)'$ and $(\NMod\to)'$ are translated as follows:
$$
\infer[(\to\cdot) \text{ two times}]{(q \BS q)^k, [\tau^-(\Pi)], (q \BS q)^m \to (q \BS q) \cdot \PMod \tau^+(A) \cdot (q \BS q)}
{(q \BS q)^k \to q \BS q & \infer[(\to\PMod)]{[\tau^-(\Pi)] \to \PMod \tau^+(A)}{\tau^-(\Pi) \to \tau^+(A)} & (q \BS q)^m \to q \BS q}
$$
$$
\infer[(\BS\to) \text{ and } (\SL\to)]
{\Delta([(q \BS q)^k, (q \BS q) \BS \NMod \tau^-(B) \SL (q \BS q), (q \BS q)^m]) \to C}
{(q \BS q)^k \to q \BS q & \infer[(\NMod\to)]{\Delta([\NMod \tau^-(B)]) \to C}{\Delta(\tau^-(B)) \to C} &
(q \BS q)^m \to q \BS q}
$$

The ``if'' part is easier: if $\tau^-(\Pi) \to \tau^+(C)$ is derivable in $\LbS$, then it is also derivable in $\LbU$.
Substitute $\U$ for $q$ (substituting formulae for variables---but not for the unit constant!---preserves derivability).
Since $q$ is a fresh variable, the substitution affects only the $(q \BS q)$ combinations introduced by $\tau^+$ and $\tau^-$,
and it is easy to see that for any formula $B$ after this substitution $\tau^+(B)$ and $\tau^-(B)$ become equivalent to $B$.
Thus, $\tau^-(\Pi) \to \tau^+(C)$ transforms to a sequent equivalent to $\Pi \to C$, therefore the latter one is derivable in $\LbU$.
\end{proof}

\section{Applications}

The first application of Theorem~\ref{Th:main} is the characterisation of the class of languages generated by
$\LbU$-grammars. 

Let $\Sigma$ be a finite alphabet. An $\LbU$-grammar over $\Sigma$ is a triple $\langle \rhd, H \rangle$, where
$H$ is an $\LbU$-formula (called the {\em target} formula), and $\rhd$ is a finite relation between letters $\Sigma$ and 
$\LbU$-formulae. J\"ager~\cite{Jaeger} gives two definitions of a word $a_1\dots a_n$ being accepted by such a grammar.
The word is {\em s-accepted,} if there exist such formulae $A_1, \dots, A_n$ that $a_i \rhd A_i$ ($i = 1,\dots, n)$ and
the sequent $A_1, \dots, A_n \to H$ is derivable in $\LbU$. The word is {\em t-accepted,} if, again, there exist $A_1, \dots, A_n$
such that $a_i \rhd A_i$, and there also exists a multiformula $\Pi$ such that $\Pi \to H$ is derivable and $A_1, \dots, A_n$
is obtained from $\Pi$ by erasing $[$ and $]$ ({\it i.e.,} it is the {\em yield} of the treelike bracketed structure $\Pi$).
The language t-generated (resp., s-generated) by the grammars defined as the set of all t-accepted (resp., s-accepted) words.

Theorem~\ref{Th:main} immediately yields the following corollary:
\begin{theorem}
The class of languages t-generated (resp., s-generated) by $\LbU$-grammars, coincides with the class of languages
t-generated (resp., s-generated) by $\LbS$-grammars.
\end{theorem}

\begin{proof}
For the inclusion in the non-trivial direction, apply $\tau^-$ to all formulae associated to letters of
$\Sigma$ by $\rhd$, and $\tau^+$ to the target formula $H$.
\end{proof}

In combination with Kanazawa's result that $\LbS$-grammars t-generate exactly the class of context-free languages,
this gives a solution for the question left open in~\cite{KanazawaArXiv}: $\LbU$-grammars also t-generate exactly all
context-free languages. We conjecture that our construction also works for the multimodal version of $\LbU$, 
solving the second open question from~\cite{KanazawaArXiv}.
The question about the class of s-generated languages appears to be still open, though much less interesting from
the linguistic point of view.

The second application is a fast algorithm for deciding derivability in $\LbS$. While the full Lambek calculus is
NP-complete~\cite{PentusNP}, Pentus in~\cite{PentusPoly} introduces a depth parameter $d$ (in the product-free
case it is just the Horn depth of $\BS$ and $\SL$) and presents an algorithm that works in polynomial time w.r.t.\ $n$
and $2^d$ (where $n$ is the length of the input sequent). Kanovich {\em et al.}~\cite{FSCD} generalise this
algorithm to $\LbS$, adding the third parameter, $b$, which is the nesting depth of $[...]$, $\PMod$, and $\NMod$. The running
time of this algorithm is bounded by a polynom of $n$, $2^d$, and $n^b$ (if there are brackets, $n^b$ absorbs $n$). Since $\tau^+$ and
$\tau^-$ alter these parameters only linearly, Theorem~\ref{Th:main} yields an algorithm also for $\LbU$ with the same running
time estimation:
\begin{theorem}
There exists an algorithm that decides derivability in $\LbU$ with running time bounded by a polynom of
$n$, $2^d$, and $n^b$.
\end{theorem}

\subsection*{Acknowledgments}

The author
was much inspired by recent and older works of Makoto Kanazawa (to which this note is a modest addition).
The author is also grateful to Lev Beklemishev, Max Kanovich, Glyn Morrill, Vivek Nigam, Mati Pentus, and Andre Scedrov for 
in-depth discussions and exchange of ideas in a very warm, friendly, and collaborative atmosphere.

\end{document}